\documentclass[11pt]{article}

\usepackage[margin=1.25in]{geometry}
\usepackage{amsmath}
\usepackage{amsfonts}
\usepackage{amsthm}
\newtheorem{theorem}{Theorem}

\newtheorem{definition}{Definition}
\usepackage{color}
\definecolor{darkgreen}{rgb}{0,0.7,0}
\newcommand{\kibitz}[2]{}
\newcommand{\cb}[1]{\kibitz{darkgreen} {[CB: #1]}}
\newcommand{\jb}[1]{\kibitz{red} {[JB: #1]}}
\newcommand{\paren}[1]{\left(#1\right)}
\newcommand{\GF}{\mathbf{GF}}
\newcommand{\NW}{\mathbf{NW}}
\newcommand{\cut}[1]{}
\def \QED {\hfill{$\Box$}}
\newenvironment{proofof}[1]{\noindent {\em Proof of #1.  }}{\QED}

\newenvironment{remindertheorem}[1]{\medskip \noindent {\bf Reminder of Theorem #1.  }\em}{}

\title{Set Families with Low Pairwise Intersection}
\author{Calvin Beideman \\ High School \\ \texttt{mathematicsfan@gmail.com} \and Jeremiah Blocki \\ Carnegie Mellon University \\ \texttt{jblocki@cs.cmu.edu}}

\begin{document}

\maketitle

\begin{abstract}
A $\paren{n,\ell,\gamma}$-sharing set family of size $m$ is a family of sets $S_1,\ldots,S_m\subseteq [n]$ s.t. each set has size $\ell$ and each pair of sets shares at most $\gamma$ elements. We let $m\paren{n,\ell,\gamma}$ denote the maximum size of any such set family and we consider the following question: How large can $m\paren{n,\ell,\gamma}$ be?  $\paren{n,\ell,\gamma}$-sharing set families have a rich set of applications including the construction of pseudorandom number generators\cite{nisanwidgerson} and usable and secure password management schemes \cite{NaturallyRehearsingPasswords}. We analyze the explicit construction of Blocki et al \cite{NaturallyRehearsingPasswords} using recent bounds \cite{sondow2009ramanujan} on the value of the $t$'th Ramanujan prime \cite{ramanujan1919proof}. We show that this explicit construction produces a $\paren{4\ell^2\ln 4\ell,\ell,\gamma}$-sharing set family of size $\paren{2 \ell \ln 2\ell}^{\gamma+1}$ for any $\ell\geq \gamma$. We also show that the construction of Blocki et al \cite{NaturallyRehearsingPasswords} can be used to obtain a {\em weak} $\paren{n,\ell,\gamma}$-sharing set family of size $m$ for {\em any} $m >0$. These results are competitive with the inexplicit  construction of Raz et al \cite{Raz:1999:ERR:301250.301292} for weak $\paren{n,\ell,\gamma}$-sharing families. We show that our explicit construction of weak $\paren{n,\ell,\gamma}$-sharing set families can be used to obtain a parallelizable pseudorandom number generator with a low memory footprint by using the pseudorandom number generator of Nisan and Wigderson\cite{nisanwidgerson}. We also prove that $m\paren{n,n/c_1,c_2n}$ must be a constant whenever $c_2 \leq \frac{2}{c_1^3+c_1^2}$. We show that this bound is nearly tight as $m\paren{n,n/c_1,c_2n}$ grows exponentially fast whenever $c_2 > c_1^{-2}$.   
\end{abstract}

\section{Introduction} \label{sec:introduction}
Informally, we define an $\paren{n, \ell, \gamma}$-sharing set family of size $m$ to be a collection of $m$ subsets of $[n]$, each of size $\ell$, no two of which have more than $\gamma$ elements in common, and we let $m\paren{n,\ell,\gamma}$ denote the maximum size of such a set family. How large can $m\paren{n,\ell,\gamma}$ be? Can we find explicit constructions of large $\paren{n, \ell, \gamma}$-sharing set families? While these combinatorial questions are interesting in their own right, these question also have numerous practical implications including the construction of pseudorandom number generators \cite{nisanwidgerson}, randomness extractors\cite{trevisan2001extractors,Raz:1999:ERR:301250.301292} and most recently usable and secure password management scheme (systematic strategies for users to create and remember multiple passwords) \cite{NaturallyRehearsingPasswords}.

\paragraph{Applications to Pseudorandom Number Generation}  A pseudorandom number generator is a function $\mathbf{G}:\{0,1\}^n \rightarrow m$ which takes a uniformly random seed $x \sim \{0,1\}^n$ of length $n$, and outputs a string $\mathbf{G}(x) \in \{0,1\}^m$ ($m \gg n$) which ``looks random." Nisan and Wigderson used a $\paren{n, \ell = O\paren{\sqrt{n}}, \gamma = \log m}$-sharing set family $\mathcal{S} = \left\{S_1,\ldots,S_m\right\}$ of size $m$ to construct pseudorandom number generators \cite{nisanwidgerson}. In particular, they define the pseudorandom number generator $\NW_{P,\mathcal{S}}\paren{x} = P\paren{x_{|S_1}}\ldots P\paren{x_{|S_m}}$, where $x_{|S_i} \in \{0,1\}^\ell$ denotes the bits of $x \in \{0,1\}^\ell$ at the indices specified by $S_i$ and $P:\{0,1\}^\ell \rightarrow \{0,1\}$ is a predicate. If the  predicate $P:\{0,1\}^\ell \rightarrow \{0,1\}$ is ``hard" for circuits of size $H_\ell\paren{P}$ to predict \footnote{Nisan and Wigderson observe that a random predicate $P$ will satisfy this property with high probability\cite{nisanwidgerson}.} then no circuit of size $H_\ell\paren{P} - O\paren{m2^\gamma}$ will be able to distinguish $\NW_{P,\mathcal{S}}\paren{x}$ from a truly random binary string of length $m$, when the seed $x \sim \{0,1\}^n$ is chosen uniformly at random. In this context, $n$ is the length of the random seed, $m$ is the number of random bits extracted and the pseudorandom number generator fools circuits of size $H_\ell\paren{P}- O\paren{m2^\gamma}$. Thus, we would like to find $\paren{n,\ell,\gamma}$-sharing set families where $n$ is small, $m$ is large (e.g., we can extract many pseudorandom bits from a small seed) and $\gamma$ is small (e.g., so that the pseudorandom bits look random to a large circuit). Nisan and Wigderson gave an explicit construction of an $\paren{\ell^2,\ell,\gamma}$-sharing set family of size $\ell^{\gamma+1}$. 

\paragraph{Applications to Randomness Extractors} Trevisan used the pseudorandom number generator of Nisan and Wigderson to construct a randomness extractor \cite{trevisan2001extractors}. A $\paren{k,\epsilon}$ randomness extractor is a function $\mathbf{Ext}: \{0,1\}^{\hat{\ell}} \times \{0,1\}^n \rightarrow \{0,1\}^m$ that takes a string $x_1 \sim D$, where $D$ is a distribution over $\{0,1\}^{\hat{\ell}}$ with minimum entropy $k$, along with a $n$ additional uniformly random bits $x_2 \sim \{0,1\}^n$ and extracts an $m$-bit string $y \in \{0,1\}^m$ that is almost uniformly random (e.g., distribution over $y \in \{0,1\}^m$ is $\epsilon$-close to the uniform distribution $U_m$ over $\{0,1\}^m$). Trevisan used the string $x_1$ to select a random predicate $P:\{0,1\}^\ell \rightarrow \{0,1\}$, and then extracted $m$ bits by running $\NW_{P,\mathcal{S}}\paren{x_2}$. Raz et al \cite{Raz:1999:ERR:301250.301292} observed that the pseudorandom number generator Nisan and Wigderson could be built using a {\em weak} $\paren{n,\ell,\gamma}$-sharing set family of size $m$, and showed how to construct {\em weak} $\paren{n,\ell,\gamma}$-sharing set family of size $m$ for {\em any} value of $m$ as long as $n \geq \lceil \frac{\ell}{\gamma}\rceil \ell$. However, their construction was not explicit.

\paragraph{Advantages of Explicit Constructions} One nice property of the Nisan Wigderson Pseudorandom number generator is that it is highly parallelizable. For each $j \in [m]$ we can compute the $j$'th bit $\NW_{P,\mathcal{S}}\paren{x}[j] = P\paren{x_{|S_j}}$ independently as long as we can quickly find the set $S_j \in \mathcal{S}$. Observe that we would need space at least $O\paren{m\ell\log n}$ to store the set family $\mathcal{S}=\left\{S_1,\ldots,S_m\right\}$, which could be a problem especially when $m$ is very large. However, if the set family has an explicit construction (e.g., there is a small circuit $C$ s.t. $C\paren{i} = S_i$ for all $i \in [m]$) then we can simply compute $\NW_{P,\mathcal{S}}\paren{x}[j] = P\paren{x_{|C(j)}}$.

\cut{The problem of finding the maximum number of subsets, no two of which have an intersection whose size exceeds a specified maximum is one which has been applied to pseudorandom number generation by Nisan and Wigderson and to password management systems by Blocki et al. Nisan and Wigderson explored the case where the size of the subsets is proportional to the square root of the size of the larger set and the maximal intersection size is proportional to the logarithm of the larger set's size. In the pseudorandom number generation context, $n$ is the number of random bits provided, $m$ is the number of pseudorandom bits that can be generated from these random bits, it is assumed that we have a function mapping $\ell$ bits to a single bit that is hard for small circuits to accurately predict, and their construction requires that $\gamma$ equals O(log(n)). We use a construction described by Blocki et al based on the Chinese Remainder Theorem to create a tighter bound than Nisan and Wigderson's for the cases they describe. In the password management context, $n$ is the number of cue-association pairs the user must memorize, $m$ is the maximum number of passwords that can be created, $\ell$ is the number of cues used in each password, and $\gamma$ is the maximum number of cues shared between passwords.\\}

\paragraph{Applications to Password Management} Recently Blocki et al \cite{NaturallyRehearsingPasswords} used $\paren{n,\ell,\gamma}$--sharing set families to develop usable and secure password management schemes. In their proposed password management scheme, Shared Cues, the user memorizes and rehearses $n$ secret stories. From these $n$ stories the user is able to create $m\paren{n,\ell,\gamma}$ different passwords. In particular, the password at each of the user's accounts is formed by appending $\ell$ of these secret stories together. A usable password management scheme should keep $n$ and $\ell$ as small as possible so that the user does not have to memorize too many stories and type too many stories when he logs into an account. $\gamma$ is a security parameter which specifies how much information one password might leak about another (e.g., if an adversary learns the user's Amazon password then he learns at most $\gamma$ of the user's stories for eBay). A secure password management scheme should keep $\gamma$ as small as possible (so that one password does not leak too much information about another password) and $\ell$ as large as possible (so that each password has high entropy). Blocki et al \cite{NaturallyRehearsingPasswords} gave a construction of $\paren{n,\ell,\gamma}$--sharing set families using the Chinese Remainder Theorem. Given pairwise coprime numbers $n_1,\ldots,n_\ell$ s.t. $n = n_1 + \ldots + n_\ell$ they construct $S_1,\ldots,S_m$ where $S_i = \{1+\sum_{k=1}^{j-1} n_k+\paren{i \mod{n_j}}~:~ j\in [\ell] \}$. They use the Chinese Remainder Theorem to prove that $\max_{i\neq j} \left|S_i \cap S_j \right|\leq \gamma$ as long as $m \leq \prod_{i=1}^{\gamma+1} n_i$.

\paragraph{Contributions} We analyze the explicit construction of Blocki et al \cite{NaturallyRehearsingPasswords} and show that it is competitive with the explicit construction of Nisan and Wigderson \cite{nisanwidgerson}. Our analysis uses recent bounds \cite{sondow2009ramanujan} on the value of the $t$'th Ramanujan prime \cite{ramanujan1919proof}. We also show that the construction of Blocki et al can be used to explicitly construct {\em weak} $\paren{n,\ell,\gamma}$-sharing set families whose size is very large. Our analysis shows that this explicit construction is competitive with the non-explicit construction of Raz et al \cite{Raz:1999:ERR:301250.301292}. We show that our explicit construction of weak $\paren{n,\ell,\gamma}$-sharing set families can be used to obtain a parallelizable pseudorandom number generator with a low memory footprint by using the pseudorandom number generator of Nisan and Wigderson\cite{nisanwidgerson}.We also prove several upper bounds on the value of $m\paren{n,\ell,\gamma}$ when $\ell$ and $\gamma$ are in a constant ratio to $n$.

\cut{ We also explore the case where both the size of the subsets and the size of their maximal intersection are proportional to the size of the larger set, and describe conditions under which the maximum number of subsets is constant as well as those under which it grows exponentially. We provide some works on the densities of relatively prime numbers which could potentially be used to strengthen our bound for the case where the size of the larger set is on the order of the square of the size of the subsets.}

\paragraph{Organization} The paper is organized as follows: We first introduce related work in Section \ref{subsec:RelatedWork}. We then introduce preliminary definitions in Section \ref{subsec:Preliminaries}. In Section \ref{sec:constructions} we analyze the construction of Blocki et. al, and state a lower bound on $m\paren{n,\ell,\gamma}$ that can be derived from it. We compare this lower bound to the construction of Nisan and Wigderson. We also show that this explicit construction yields a good {\em weak} $\paren{n,\ell,\gamma}$-sharing set family. In Section \ref{sec:Pseudorandomness} we explain how the explicit construction of Blocki et al \cite{NaturallyRehearsingPasswords} can be used to obtain a highly parallelizable pseudorandom number generator with a low memory footprint. In Section \ref{sec:UpperBounds} we explore some cases where $\ell$ and $\gamma$ are in a constant ratio to $n$ and prove an upper bound on $m\paren{n,\ell,\gamma}$ as $n$ grows large. We show that our upper bounds are nearly tight. We conclude in Section \ref{sec:OpenQuestions} by discussing cases that do not meet the conditions for any of our bounds, and hypotheses about how our bounds could be made stronger.

\subsection{Related Work} \label{subsec:RelatedWork}
The problem of finding maximally sized $\paren{n,\ell,\gamma}$--sharing set families was considered at least as early as 1956 by Paul Erd\H{o}s and Alfr\'{e}d R\'{e}nyi \cite{erdos1956some}, and applications of some of these families may have been considered by Euler \cite{euler1782recherches}. Erd\H{o}s explored properties of these families several times \cite{erd6s1963limit} \cite{erdos1985families}, and R{\"o}dl built on his work \cite{rodl1985packing}.

$\paren{n,\ell,\gamma}$--sharing set families were rediscovered by Nisan and Wigderson \cite{nisanwidgerson}, who used them to design a pseudorandom number generator. Trevisan showed how to use $\paren{n,\ell,\gamma}$--sharing set families to construct pseudorandom extractors \cite{trevisan2001extractors}. Extractors are algorithms that transform weakly random sources into a uniformly random source. Raz et al \cite{Raz:1999:ERR:301250.301292} improved on Trevisan's pseudorandom extractors by introducing a weakened notion of $\paren{n,\ell,\gamma}$--sharing set families. They require that the set family $S_1,\ldots, S_m \subseteq [n]$ satisfies $\left|S_i\right| = \ell$ and $\sum_{j < i} 2^{\left| S_i \bigcap S_j\right|} \leq 2^\gamma \paren{m-1}$ for all $i \in [m]$ (instead of $\left|S_i \bigcap S_j\right| \leq \gamma$). Observe that every $\paren{n,\ell,\gamma}$--sharing set family also satisfies these weaker requirements. Using this relaxed definition Raz et al \cite{Raz:1999:ERR:301250.301292} showed how to extract a uniformly random string $y \in \{0,1\}^{k}$ using at most $O\paren{\log^3 n}$ bits of information given a string $x \in \{0,1\}^n$ chosen at random from a distribution $D$ with minimum entropy $k$. To obtain their results they show how to construct very large {\em weak} $\paren{n,\ell,\gamma}$--sharing set families. However, their construction is not explicit. We use the construction of Blocki et al to obtain an explicit construction of large {\em weak} $\paren{n,\ell,\gamma}$--sharing set families.

 Blocki et al \cite{NaturallyRehearsingPasswords} proposed a construction of $m\paren{n,\ell,\gamma}$--sharing set families based on the Chinese Remainder Theorem. In their analysis of their construction they focused on parameters that were appropriate for the context of password management (e.g., $\ell = 4,\gamma = 1, n=43$). We extend their analysis to include a broader range of parameters. 
Our analysis uses recent results of Sondow \cite{sondow2009ramanujan}, who  provided a (nearly) asymptotically tight bound on the value of the $t$'th Ramanujan prime \cite{ramanujan1919proof}. We show that the construction of Blocki et al \cite{NaturallyRehearsingPasswords} yields a larger $\paren{n,\ell,\gamma}$-sharing set family than the construction of Nisan and Widgerson \cite{nisanwidgerson} with equivalent values of $n$ and $\gamma$ (though the value of $\ell$ is slightly smaller).

\jb{Extending Blocki's analysis and analyzing their construction within the weaker definition/requirements. Providing new upper bounds when l = O(n) and gamma = O(n). Claiming our construction is slightly worse than that of Raz et al, but that ours is explicit and theirs is not. Explicit constructions have the advantage of being able to compute the ith bit independently.}

\section{Preliminaries} \label{subsec:Preliminaries}

We begin by formally defining an $(n, \ell, \gamma)$--sharing set family (Definition \ref{def:main}).

\begin{definition}\label{def:main} An $(n, \ell, \gamma)$--sharing set family $S_1, \ldots, S_m\subseteq[n]$ of size $m$ satisfies the following conditions: 
(1) $\forall i \in [m]. \left| S_i\right| = \ell$, and (2) $\forall 1\leq i < j \leq m.~ \left|S_i \bigcap S_j\right| \leq \gamma$.
We use $m(n, \ell, \gamma)$ to denote the maximum value of $m$ such that there exists an $n, \ell, \gamma$ sharing set family  of size $m$. We say that a set family $S_1, \ldots, S_m\subseteq[n]$ is {\em explicitly constructible} if there is a circuit $C$ of size $O\paren{n}$ that computes $C(i) = S_i$  for each $i \in [m]$. 
\end{definition}

Nisan and Wigderson referred to these families as $(k,m)$-designs \cite{nisanwidgerson}. We follow the notation of Blocki et al \cite{NaturallyRehearsingPasswords}. The construction of Blocki et al \cite{NaturallyRehearsingPasswords} relies on the Chinese Remainder Theorem. To analyze their construction we will be interested in finding a large set $S = \{t_1,\ldots,t_\ell\}$ of integers such that $S$ has size $\ell$, the numbers in $S$ are pairwise coprime, $\sum_{i=1}^\ell t_i \leq n$ and each $t_i\geq \frac{n}{2\ell}$. We will rely on recent results on prime density. 

\begin{definition} \label{def:pi}
$\pi (t)$ indicates the number of prime numbers less than or equal to $t$.
$\pi\pi (t)$ indicates the maximum $|S|$ such that $S\subseteq \left\{\lceil\frac{t}{2}\rceil, ... , t\right\}$ and $\forall i \neq j \in S. \mathbf{GCD}\paren{i, j} = 1$. \end{definition}

We are particularly interested in lower bounding the value $\pi\pi(x)$. Clearly,  $\pi\pi(x) \geq \pi(x)-\pi(x/2)$. As it turns out this lower bound is nearly tight (see Theorem \ref{thm:pipibound}). We can bound $\pi(x)-\pi(x/2)$ using Ramanujan primes.

\begin{definition} \cite{ramanujan1919proof} \label{def:ramanujan}
The t'th Ramanujan Prime is the smallest integer $R_t$ s.t. $\pi(x)-\pi(x/2) \geq t$ for all $x \geq R_t$.
\end{definition}
Allowing $n$ to equal at least $\ell R_\ell$ guarantees that $\left\{ \frac{n}{2\ell}, \frac{n}{\ell} \right\}$ contains at least $\ell$ primes which will satisfy the conditions of the Blocki conjecture. Sondow's bounds on Ramanujan primes (see Theorem \ref{thm:ramanujan}) allow us to express this bound on $n$ as an elementary function. 

\subsection{Pseudorandom Number Generators and Randomness Extractors}
\jb{Want to have relevant definitions from Nisan-Wigderson, Trevisan, and Raz et al. What does it mean for a circuit to distinguish distributions with advantage $\epsilon$. What does it mean for a predicate to be hard for circuits of size $s$? What is a randomness extractor? Can copy definitions from other papers. Also formal definition of weaker set families. I will add text explaining why the weaker set families are relevant.}

Before we formally define a pseudorandom number generator we first define a pseudorandom distribution $X$ over $\{0,1\}^m$. Informally, definition \ref{def:psuedorandomdist} say that distribution is pseudorandom a distribution that `appears' random to any `small enough' circuit. Given a circuit $C$ we use \[\mathbf{Adv}_C\paren{X} = \left| \Pr_{x \in X} [C(x) = 1] - Pr_{x\in U_m} [C(x) = 1] \right| \, \]
to denote the advantage of $C$ at predicting whether $x$ was drawn from the distribution $X$ or from $U_m$, where $U_m$ is the uniform distribution over $\{0,1\}^m$. The distribution $X$ `appears' random to a circuit $C$ if $\mathbf{Adv}_C\paren{X}$ is small.

\begin{definition}\label{def:psuedorandomdist}
A distribution $X$ over $\{0,1\}^m$ is said to be $(s,\epsilon)$-pseudorandom if, given any circuit $C$ (taking $m$ inputs) of size at most $s$,
$\mathbf{Adv}_C\paren{X}   \leq \epsilon$.

\end{definition}

Given a distribution $X$ over $\{0,1\}^n$ and a function $G:\{0,1\}^n\rightarrow \{0,1\}^m$ we use $G(X)$ to denote the distribution over $\{0,1\}^m$ induced by $G$. Informally, a function $G:\{0,1\}^n\rightarrow \{0,1\}^m$ is pseudorandom if it induces a pseudorandom distribution. 

\begin{definition}\label{def:prg}
Let $\{G_n\}_{n\in N}$ be a family of functions such that $G_n : \{0,1\}^n \rightarrow \{0,1\}^m$. We say the family is a $(s,\epsilon)$-pseudorandom number generator if G is computable in time $2^{O(n)}$, and $G(U_n)$ considered as a distribution is $(s, \epsilon)$-pseudorandom.
\end{definition}

Nisan and Wigderson \cite{nisanwidgerson} show how to construct a pseudorandom number generator $G:\{0,1\}^n\rightarrow\{0,1\}^m$ using any $\paren{n,\ell,\gamma}$-sharing set family of size $m$. Their construction assumes the existence of a predicate $f:\{ 0,1 \} ^\ell \rightarrow \{ 0,1 \}$ that is hard for `small' circuits to predict. 

\cb{Nisan Wigderson says the error needs to be less than $\frac{\epsilon}{2}$, the notes simply use $\epsilon$. Does it matter for our purposes?}
\begin{definition}\label{def:hardfunction}
Let $f:\{ 0,1 \} ^\ell \rightarrow \{ 0,1 \}$ be a boolean function. We say that $f$ is $(s, \epsilon)$-hard if for any circuit $C$ of size $s$, $ \left|\Pr_{x\sim\{0,1\}^\ell}\left[C(x)=f(x) \right]-\frac{1}{2} \right| \leq \epsilon.$ 
\end{definition}

Observe that a random function will fool all small circuits with high probability \footnote{The argument is straightforward. Fix any circuit $C$. A random function $f:\{ 0,1 \} ^\ell \rightarrow \{ 0,1 \}$ will satisfy $\mathbf{Adv}_C\paren{f\paren{U_\ell}} \leq \epsilon$ with very high probability by Chernoff bounds. We can then apply union bounds to argue that a random $f$ will satisfy $\max_{C \in \mathcal{C}} \mathbf{Adv}_C\paren{X} \leq \epsilon$ for any sufficiently small class $\mathcal{C}$ of circuits. }. Following, Nisan and Wigderson we use $H(f)$ to denote the hardness of a function $f$. 

\begin{definition}\label{def:hf}
Let $f:\{ 0,1\} ^* \rightarrow \{ 0,1 \}$ be a boolean function and let $f_\ell$ be the restriction of $f$ to strings of length $\ell$. The $hardness$ of $f$ at $\ell$, $H_f(\ell)$ is defined to be the maximum integer $h_\ell$ such that $f_\ell$ is $(1/h_\ell, h_\ell)-hard$.
\end{definition}

Raz et al \cite{Raz:1999:ERR:301250.301292} showed that the Nisan-Wigderson pseudorandom number generator works even if the family of sets $S_1, ... , S_m $ only satisfies the weaker condition from definition \ref{def:weakfamily}. Observe that any $(n,\ell,\gamma)$-sharing set family is also a {\em weak} $(n,\ell,\gamma)$-sharing set family, but the converse is not necessarily true. We also note that as $m$ increases the requirement $\sum_{j<i} 2^{\left|S_i \bigcap S_j\right|} \leq 2^\gamma (m - 1)$ becomes increasingly lax. This allows us to construct arbitrarily large weak $(n,\ell,\gamma)$-sharing families.

\begin{definition}\label{def:weakfamily}
A family of sets $S_1, ... , S_m \subset [n]$ is a weak $(n,\ell,\gamma)$-sharing set family if (1) $\forall i\in[m]$. $\left|S_i\right| = \ell$, and (2) $\forall i\in [m]$.$ \sum_{j<i} 2^{\left|S_i \bigcap S_j\right|} \leq 2^\gamma (m - 1)$. 
\end{definition}

\section{Constructions} \label{sec:constructions}
Nisan and Wigderson \cite{nisanwidgerson} gave an explicit construction of $\paren{\ell^2,\ell,\gamma}$-sharing set families of size $m = \ell^{\gamma+1}$ for any prime power $\ell$. Given a polynomial $p(x)$ with coefficients in $\GF(\ell)$, the finite field of size $\ell$, they define the set $S_p = \left\{\paren{x,p(x)}~\vline~x\in \GF\paren{\ell}\right\}$. The family $\mathcal{S} = \left\{S_p~\vline~\mbox{$p$ has degree $\leq \gamma$}\right\}$ is $\paren{\ell^2,\ell,\gamma}$-sharing and has size $m = \left|\mathcal{S} \right| = \ell^{\gamma+1}$. Given pairwise coprime numbers $n_1< \ldots <n_\ell$ Blocki et al \cite{NaturallyRehearsingPasswords} provided an explicit construction of $\paren{\sum_{i=1}^\ell n_i,\ell,\gamma}$-sharing families. Given an integer $i \geq 0$ they define the set $S_i = \{1+\sum_{k=1}^{j-1} n_k+\paren{i \mod{n_j}}~:~ j\in [\ell] \}$. They show that the family $\mathcal{S} = \left\{S_i~\vline~0 \leq i < \prod_{j=1}^{\gamma+1} n_i \right\}$ is an $\paren{\sum_{i=1}^\ell n_i,\ell,\gamma}$-sharing set family of size $\prod_{j=1}^{\gamma+1} n_i$.
  
\jb{The proof of Theorem \ref{thm:primedistribution} is based on the following result of Blocki et al \cite{NaturallyRehearsingPasswords}. State Theorem. A couple of sentence comparing Theorem \ref{thm:primedistribution} to Nisan Wigderson construction \cite{nisanwidgerson}. Add a couple sentences explaining blocki et al construction. Note that it is explicit. }
The proof of Theorem \ref{thm:primedistribution} is based on the following result of Blocki et al \cite{NaturallyRehearsingPasswords}. We take advantage of Sondow's results on prime density \cite{sondow2009ramanujan} to compare the Blocki et al construction to the construction of Nisan and Wigderson.
\begin{theorem} \cite{NaturallyRehearsingPasswords} \label{thm:CRTImprovement}
Suppose that $n_1 < \ldots < n_\ell$ are pairwise co-prime then there is a  $(\sum_{i=1}^\ell n_i,\ell,\gamma)$--sharing set system of size $m = \prod_{i=1}^\gamma n_i $. Furthermore, this set family has an explicit construction.
\end{theorem}

\begin{theorem}\cite{sondow2009ramanujan}\label{thm:ramanujan}
For all $t \geq 1$ the following bound holds $2t \ln t < R_t < 4t \ln 4t$.
\end{theorem}

\begin{theorem} \label{thm:primedistribution}
$\forall n \geq 4\ell^2\ln 4\ell$, $m(n, \ell, \gamma) \geq (2\ell\ln 2\ell)^{\gamma + 1}$. Furthermore, this set family is explicitly constructible.
\end{theorem}

\begin{proof} 
Theorem \ref{thm:ramanujan} due to Sondow \cite{sondow2009ramanujan} shows that there will always be at least $\ell$ primes $p_1,\ldots, p_\ell$ between $2\ell\ln 2\ell$ and $4\ell \ln 4\ell$. We have $\sum_{i=1}^{\ell} p_i \leq \ell (4\ell \ln 4\ell) \leq n$. Note that $\prod_{i=1}^{\gamma +1} p_i \geq  (2\ell\ln 2\ell)^{\gamma+1}$. It follows from Theorem \ref{thm:CRTImprovement} that $m\paren{n,\ell,\gamma} \geq \paren{2\ell\ln 2\ell}^{\gamma+1}$.
\end{proof}

Note that the construction of Blocki et al only requires relatively prime numbers. So the results from theorem \ref{thm:primedistribution} could be improved by including non-prime values. However, theorem \ref{thm:pipibound} implies that these improvements will not be particularly significant. \jb{Is there a clever one sentence summary of the intuition?}

\newcommand{\thmpipibound}{$\forall n \in \mathbb{Z}^+$. $\pi\pi(n) \leq \pi(n) - \pi(\frac{n}{2}) + \pi(\sqrt{n})$. }
\begin{theorem}\label{thm:pipibound}
\thmpipibound
\end{theorem}
\begin{proof}
Let $S \subseteq \left\{\lceil \frac{n}{2} \rceil,\ldots,n\right\}$ be a set of coprime numbers of maximum size. Observe that each prime number $p \in [n]$ is a factor of at most one number in $S$. Without loss of generality we can assume that each of the primes between $n$ and $\frac{n}{2}$  are contained in $S$ (if $p \notin S$ then, because $S$ is of maximum size, we must have some $t=pq\in S$, but in this case we can simply replace $t$ with $p$). The number of primes between $n$ and $\frac{n}{2}$ is $\pi(n) - \pi(\frac{n}{2})$, and all of these integers are relatively prime to each other and to every other number in the range $[n]$. All other numbers in $S$ must have at least two prime factors, and at least one of them must be less than or equal to $\sqrt{n}$. Since each prime factor less than or equal to $\sqrt{n}$ can be used at most once, for the members of $S$ to remain pairwise relatively prime, at most $\pi (\sqrt{n})$ non-primes can be included in the set, each containing a single prime factor less that $\sqrt{n}$.
\end{proof}

\paragraph{Comparison.} To compare the constructions of Blocki et al \cite{NaturallyRehearsingPasswords} and Nisan and Wigderson \cite{nisanwidgerson} we set $n = 4\ell'^2  \ln 4\ell'$ and we set $\ell = \sqrt{4 \ell'^2 \ln 4 \ell'}$. The construction of Nisan and Wigderson gives use $m\paren{n,\ell,\gamma} \geq \ell^{\gamma+1} = \paren{2\ell'\sqrt{\ln 4\ell'}}^{\gamma+1}$, while the construction of Blocki et al \cite{NaturallyRehearsingPasswords} gives us $m\paren{n,\ell',\gamma} \geq  \paren{2\ell'\ln 2\ell'}^{\gamma+1} > \paren{2\ell'\sqrt{\ln 4\ell'}}^{\gamma+1} $. However, $\ell' < \ell$ so the construction of Blocki et al has a smaller $\ell$.

\cut{\begin{theorem} \label{thm:nisanbound} \cite{nisanwidgerson}$\forall m, l \in \mathbb{z}$ such that $log(m) \leq l \leq m$ $\exists$ a set family $S_1, ... S_m$ where $n=O(l^2)$ and $\gamma = log(n)$. Equivalently: $l \leq m(O(l^2), l, log(n)) \leq e^\ell$\end{theorem}

Let $\ell ' = 2\ell \sqrt{4 \ln \ell}$. Nisan and Wigderson's bound then gives that $\ell ' \leq m(O(\ell '^2), \ell ', log(n)) \leq e^{\ell '}$, so $ 2\ell \sqrt{4 \ln \ell} \leq m(O(4 \ell \ln 4\ell), $}

\subsection{Constructing Weak $\paren{n,\ell,\gamma}$-sharing set families} \label{subsec:weakConstruction}
In this subsection we show that the techniques of Blocki et al \cite{NaturallyRehearsingPasswords} yield an explicit construction of weak $\paren{n,\ell,\gamma}$-sharing set families of arbitrary size $m$. Our main results are stated in Theorem \ref{thm:weakConstruction}. 

\begin{theorem} \label{thm:weakConstruction}
For all $m$ there is an explicitly constructible weak $\paren{4 \ell^2 \ln 4\ell,\ell,\gamma}$-sharing set family of size $m$ as long as $2^\gamma \geq \paren{1+\frac{1}{-1+\ln 2 \ell}}$. Furthermore, this set family is explicitly constructible. 
\end{theorem}
\begin{proof}
Let $m$ be given. We use the explicit construction of Blocki et al \cite{NaturallyRehearsingPasswords}. By Theorem \ref{thm:ramanujan} we can find $\ell$ primes such that $2 \ell \ln 2 \ell < p_1 < \ldots < p_\ell < 4\ell \ln 4\ell$. In particular, we let $S_i = \left\{1+\sum_{k=1}^{j-1} p_k + \paren{i \mod{p_j}}~\vline~j \in [\ell]\right\}$. Now for $i \in [m]$ we have
\begin{eqnarray*}
\sum_{j < i} 2^{\left|S_i\cap S_j\right|} &=& \sum_{k=0}^\infty 2^k \left|\left\{j~\vline~j<i \wedge \left|S_i \cap S_j \right|=k \right\} \right| \leq \sum_{k=0}^\infty 2^k \left|\left\{j~\vline~j<i \wedge \left|S_i \cap S_j \right|\geq k \right\} \right|  \\
&\leq& \sum_{k=0}^\infty 2^k {\ell \choose k} \frac{i-1}{\prod_{j=1}^k p_i} \leq \sum_{k=0}^\infty 2^k {\ell \choose k} \frac{i-1}{\paren{2\ell \ln 2\ell}^k} \\
&\leq&  \sum_{k=0}^\infty  \frac{i-1}{\paren{\ln 2\ell}^k} \leq \paren{i-1} \paren{\frac{\ln 2\ell}{-1+\ln 2\ell}} \leq \paren{m-1} 2^\gamma \\
\end{eqnarray*}
\end{proof}

Raz et al gave a randomized construction of weak $\paren{\left\lceil {\frac{\ell}{\gamma}}\right\rceil \cdot \ell,\ell,\gamma}$-sharing set families for any $m,\gamma > 0$.  While they showed that their construction could be derandomized, their construction is not explicit (e.g., the construction of $i$'th subset $S_i$ is dependent on the sets $S_1,\ldots,S_{i-1}$). Our analysis shows that the construction of Blocki et al \cite{NaturallyRehearsingPasswords} is competitive with the construction of Raz et al \cite{Raz:1999:ERR:301250.301292} though the value of $n$ is slightly larger. 

\cut{
\begin{theorem}\cite{Raz:1999:ERR:301250.301292}\label{thm:lemma15}
For every $m$, $\ell$, $\gamma \in \mathbb{n}$, there exists a weak $(n, \ell, \gamma)$-sharing set family $S_1, ... S_m \subset [n]$ with 
\[n = \left\lceil {\frac{\ell}{\gamma}}\right\rceil \cdot \ell \ .\]
Moreover, such a family can be found in time $poly(m,n)$.
\end{theorem}
}

\section{Parallel Pseudorandom Number Generators} \label{sec:Pseudorandomness}
\jb{add definitions, statement of theorems, explanation of Nisan Wigderson generator and translate notation to be consistent with ours. Want to compare results with construction of Raz et al. Our bound is nearly as good, and our construction is explicit. Yields new interesting results in PRGs}

Nisan and Wigderson proved that if $\gamma = \log m$, $\mathcal{S}$ is a $\paren{n,\ell,\gamma}$-sharing set family and $H_f\paren{\ell} \geq 2m^2$ that their construction $\NW_{f,\mathcal{S}}$ is a $\paren{m^2,\frac{1}{m}}$ pseudorandom number generator. In particular, Theorem \ref{thm:lemmafour} implies that if $D$ is a circuit of size $\left|D\right|\leq m^2$ that distinguishes $\NW_{f,\mathcal{S}}\paren{U_n}$ from $U_m$ with advantage $\mathbf{ADV}_D\paren{\NW_{f,\mathcal{S}}\paren{U_n}} \geq \frac{1}{m}$ then there exists a circuit $C$ of size $\left|C\right| \leq 2m^2$ which predicts $f(x)$ with advantage $\mathbf{ADV}_C\paren{f\paren{U_\ell}} \geq \frac{1}{2m^2}$. This contradicts the definition of $H_f\paren{\ell}$. Raz et al \cite{Raz:1999:ERR:301250.301292} observed that it suffices for $\mathcal{S}$ to be a weak $\paren{n,\ell,\gamma}$-sharing set family. If we let $\mathcal{S}_m$ denote the explicitly constructible weak $\paren{4\ell^2\ln 4\ell,\ell, \gamma}$-sharing set family of size $m$ from Section \ref{subsec:weakConstruction} then for any $m >0$ $\NW_{f,\mathcal{S}_m}$ is a $\paren{m^2,\frac{1}{m}}$ pseudorandom number generator with seed length $4\ell^2\ln 4\ell$ assuming that $H_f\paren{\ell} \geq 2m^2$. Because $\mathcal{S}_m$ is explicitly constructible we can compute each bit $\NW_{f,\mathcal{S}_m}\paren{x}[i] = f\paren{x_{|S_i}}$ independently.

\begin{theorem} \cite{nisanwidgerson,Raz:1999:ERR:301250.301292} \label{thm:lemmafour}
Let $f: \{ 0,1 \} ^\ell \rightarrow \{ 0,1 \}$ be a boolean function and $\mathcal{S} = \left\{S_1, ... , S_m\right\}$ be an weak $\paren{n,\ell, \gamma}$-sharing set family. Suppose $D: \{ 0,1 \} ^m \rightarrow \{ 0,1 \}$ is such that $\mathbf{ADV}_D\paren{\NW_{f,\mathcal{S}}\paren{U_n}} > \epsilon$, then there exists a circuit $C$ of size $\left|C\right|\leq \left|D\right| + O\paren{\max_{j \in [m]} \sum_{i<j} 2^{\left|S_i \bigcap S_j \right| } m}$ such that $ \left|\Pr_{x\sim\{0,1\}^\ell}\left[C(x)=f(x) \right]-\frac{1}{2} \right|  \geq \frac{\epsilon}{m}$
\end{theorem}

\cut{
\begin{definition}\label{def:minentropy}
A distribution $X$ on $ \{ 0,1 \} ^n$ is said to have min-entropy $k$ if for all $x \in \{ 0,1 \} ^n$, $\Pr [X=x] \leq 2^{-k}$.
\end{definition}

\begin{definition}
Two distributions $X$ and $Y$ on a set $S$ are said to have statistical difference $\epsilon$ if 
\[ \max_D |\Pr [D(X)=1] - \Pr [D(Y)=1] | = \epsilon \] 
\end{definition}}

\section{Upper Bounds} \label{sec:UpperBounds}
Our main result in this section is Theorem \ref{thm:constantm}.  We prove that $m(n,\ell,\gamma) = c_1$ whenever $\ell = \frac{n}{c_1}$ and $\gamma = c_2n$ provided that $c_2$ is sufficiently small. Blocki et al proved that $m(n,\ell,\gamma) \leq \frac{{n \choose {\gamma + 1}}}{{\ell \choose {\gamma +1 }}}$. We note that this bound is far from tight whenever $\ell$ is large. For example, if $c_1 = 2$ and $c_2 = \frac{1}{10}$ then the upper bound of Blocki et al ${{n \choose \frac{n+10}{10} }} \Big/ {{\frac{n}{2} \choose \frac{n+10}{10} } }$ grows exponentially with $n$. By contrast, Theorem \ref{thm:constantm} implies that $m\paren{n,n/2,n/10} = 2$.

\newcommand{\theoremconstantm}{$\forall$ $0<c_2<1, n, c_1 \in \mathbb{N} $ such that $c_1 \vert n$. $m(n, \frac{n}{c_1}, c_2n) = c_1$ iff $c_2 < \frac{2}{c_1^3 + c_1^2}$.}
\begin{theorem} \label{thm:constantm} \theoremconstantm \end{theorem}

The proof of Theorem \ref{thm:constantm} can be bound in the appendix. We instead prove an easier result here. Theorem \ref{thm:upperboundonm} upper bounds $\lim_{n\rightarrow \infty} m(n,\ell,\gamma)$ when $\ell$ is in a constant ratio to $n$ and $\gamma$ is small. Theorem \ref{thm:upperboundonm} holds because the $k$'th set $S_k$ must use $cn-(k-1)\gamma$ new elements (elements that are not in $\bigcup_{i=1}^{k-1} S_i$).
\newcommand{\theoremupperboundonm}{$\forall$ $\gamma_c$, $0<c<1$ such that $cn \in \mathbb{N}$. $m(n, cn, \gamma_c) \rightarrow \lfloor \frac{1}{c} \rfloor$ as $n \rightarrow \infty$.}
\begin{theorem} \label{thm:upperboundonm} \theoremupperboundonm \end{theorem}

\begin{proof} Let $\ell = cn$ and let  $\tau  \in \mathbb{N}$ be an integer such that $\tau > \lfloor \frac{1}{c} \rfloor$. The first set will contain $\ell$ elements. The second set can share at most $\gamma$ of them, so the second set must contain at least $\ell - \gamma$ previously unused elements. Therefore the union of the first two sets must contain at least $2\ell - \gamma$ elements. In a similar manner, the $kth$ set must contain at least $\ell - (k-1)\gamma$ new elements, therefore, 
\begin{equation} \label{eqn:thmInequality} k\ell - \frac{(k-1)k\gamma}{2} \leq \left|\bigcup_{i=1}^k S_i\right|  \leq n \ . \end{equation}
Assume for contradiction that $\limsup_{n \rightarrow \infty} m(n, cn, \gamma_c) = \tau$.Then we have 
\begin{eqnarray*}
\lim_{n \rightarrow \infty} \left(n-\tau \ell+ \frac{(k-1)k\gamma}{2}\right) &=& \lim_{n \rightarrow \infty} \left( n-\tau cn + \frac{(k-1)k\gamma}{2} \right) \\
&=& \lim_{n \rightarrow \infty} \left(n \left(1-c\tau \right)\right) \\
&=& - \infty \ .
\end{eqnarray*}
This contradicts equation \ref{eqn:thmInequality}.
\end{proof}

We also show that the upper bound from Theorem \ref{thm:constantm} is nearly tight. In particular, when $\gamma = c_2n$ for a slightly larger constant $c_2$ then $m(n,\ell,\gamma)$ is exponentially large. Theorem \ref{thm:chernoffapplication} lower bounds the values of $c_2$ for which $m(n,\ell,\gamma)$ is exponentially large.

The full proof of Theorem \ref{thm:chernoffapplication} is found in the appendix.  We demonstrate the existence of an $(n, \ell, \gamma)$--sharing set family of exponential size by showing that the probability of obtaining such a set family through random selection is non-zero. Our proof uses the following randomized construction of an $(n, \ell, \gamma)$--sharing set family. Independently choose random integers $r_i^j$ each in the range $0\leq r_i < c_1$ for $i \in \{0,\ldots,\ell-1\}$ and $j \in [m]$. Let $S_j = \bigcup\limits_{i=0}^{\ell-1}\{ ic_1 + r_i^j\}$.  We use standard concentration bounds due to Chernoff \cite{chernoff1952measure} to show that $\left| S_j \bigcap S_j\right| \leq \gamma$ with high probability, and then we union bounds to argue that the entire set family is $(n, \ell, \gamma)$--sharing with non-zero probability.

\newcommand{\thmChernoffApplication}{ $\forall$ $c_2>0, n, c_1 \in \mathbb{N}$ such that $c_1\vert n$. $m(n, \frac{n}{c_1}, c_2n) > exp(O(n))$ if $c_2 > \frac{1}{c_1^2} + \epsilon$.}
\begin{theorem} \label{thm:chernoffapplication}
\thmChernoffApplication
\end{theorem}

Blocki et al \cite{NaturallyRehearsingPasswords} observed that $m\paren{n,\gamma+1,\gamma} ={n \choose \gamma+1} $ whenever $n\geq \gamma+1$. We observe that in general $m\paren{n,\ell,\gamma} \geq m\paren{n,\ell+1,\gamma}$ whenever $\ell \geq \gamma+1$ \footnote{Suppose that $\ell \geq \gamma+1$ and we have an $\paren{n,\ell+1,\gamma}$-sharing set family $S_1,\ldots,S_m\subseteq[n]$ of size $m$. We can form a $\paren{n,\ell,\gamma}$-sharing set family $S_1',\ldots,S_m'\subseteq [n]$ by picking some element $s_i \in S_i$ setting $S_i'=S_i-\{s_i\}$ for each $i \in [m]$. Observe that this argument does not apply whenever $\ell = \gamma$ because then we might have $S_i' = S_j'$ for $i \neq j$. }. This implies that whenever $n/2\geq \gamma+1$ we have \[ \max_{\ell \geq \gamma} m\paren{n,\ell,\gamma} = m\paren{n,\gamma+1,\gamma} = {n \choose \gamma+1} \ , \]
and whenever $\gamma \geq n/2$ we have $\max_{\ell \geq \gamma} m\paren{n,\ell,\gamma} = m\paren{n,\gamma,\gamma} = {n \choose \gamma}$.
Clearly, the inequality $m\paren{n,\ell,\gamma} \geq m\paren{n,\ell,\gamma+1}$ also holds. Both of these inequalities also hold for weak $\paren{n,\ell,\gamma}$-sharing set families.

\section{Open Questions} \label{sec:OpenQuestions}
We conclude with some open questions. 

We have shown that the explicit construction of Blocki et al \cite{NaturallyRehearsingPasswords} can be used with the weaker requirements of Raz et al \cite{Raz:1999:ERR:301250.301292} to create weak $\paren{n,\ell,\gamma}$-sharing set families of arbitrarily large size. Our analysis uses a number of potentially loose bounds, however, so it is possible that a better analysis of the Blocki et al construction for weak set families could improve our requirements on the parameters. Also of interest is whether there is another explicit construction that would perform better than the Blocki et al construction.

We have shown that the value $m\paren{n,n/c_1,nc_2}$ is constant whenever $c_2 \leq \frac{2}{c_1^3 + c_1^2}$. Furthermore, we showed that whenever $c_2 > \frac{1}{c_1^2}$, $m(n,n/c_1,nc_2)$ grows exponentially. How does $m(n,n/c_1,nc_2)$ grow whenever $c_2 \in \left[\frac{2}{c_1^3 + c_1^2} ,\frac{1}{c_1^2}\right]$? 

We have shown that $\pi\pi(n)$ never exceeds $\pi(n) - \pi(\frac{n}{2}) + \pi(\sqrt{n})$. We hypothesize that $\pi\pi(n) = \pi(n) - \pi(\frac{n}{2}) + \pi(\sqrt{n})$ for all $n \geq 55$. A simple method to select a maximally-sized set of relatively prime integers is to take the square of each prime between $\sqrt{\frac{n}{2}}$ and $\sqrt{n}$, and the product of the $j$\rq{}th prime less than $\sqrt{\frac{n}{2}}$ and the $k$\rq{}th prime greater than $\sqrt{n}$, for $j$ from $1$ to $\pi(\sqrt{n})$ and $k=j$ unless this would make the product less than $\frac{n}{2}$ in which case k is chosen to be the minimum value greater than the previous k so that the product is great than $\frac{n}{2}$. With the aid of a computer we have shown this equation true for all $n$ from 1 to 100,000, except for 51, 52, 53, and 54.

\jb{Better explicit construction of weak $\paren{n,\ell,\gamma}$-sharing set families? Tighter analysis the construction of Blocki et al \cite{NaturallyRehearsingPasswords} in terms of the weaker requirements of Raz et al \cite{Raz:1999:ERR:301250.301292}.}

\bibliographystyle{alpha}
\bibliography{setsharing}

\section{Missing Proofs}
\cut{
\begin{remindertheorem}{\ref{thm:pipibound}} \thmpipibound
\end{remindertheorem}

\begin{proofof}{theorem \ref{thm:pipibound}} 
Let $S \subseteq [n]\backslash[\lceil \frac{n}{2}\rceil-1]$ be a set of coprime numbers of maximum size. Observe that each prime number $p \in [n]$ is a factor of at most one number in $S$. Without loss of generality we can assume that each of the primes between $n$ and $\frac{n}{2}$  are contained in $S$ (if we have $t=pq\in S$ where $p \in \left[\frac{n}{2},n\right]$ is prime then we could simply replace $t$ with $p$). The number of primes between $n$ and $\frac{n}{2}$ is $\pi(n) - \pi(\frac{n}{2})$, and all of these integers are relatively prime to each other and to every other number in the range $[n]$. All other numbers in $S$ must have at least two prime factors, and at least one of them must be less than or equal to $\sqrt{n}$. Since each prime factor less than or equal to $\sqrt{n}$ can be used at most once, for the members of $S$ to remain pairwise relatively prime, at most $\pi (\sqrt{n})$ non-primes can be included in the set, each containing a single prime factor less that $\sqrt{n}$.
\end{proofof}
}
\begin{remindertheorem}{\ref{thm:constantm}} \theoremconstantm \end{remindertheorem}

\begin{proofof}{theorem \ref{thm:constantm}} Suppose that for some valid $n, c_1, c_2$ there is an $(n, \ell, \gamma)$--sharing set family of size $c_1 + 1$. By equation \ref{eqn:thmInequality}, the number of elements used by such a set family must be at least: 
\begin{equation} \label{eqn:smallestgrowingm} (c_1 + 1)\ell - \frac{c_1 (c_1 + 1)\gamma}{2} \leq n \end{equation}
Taking advantage of the fact that $\ell=\frac{n}{c_1}$ and $\gamma = c_2n$, the inequality can be simplified:
\begin{eqnarray*}
n+\ell-\frac{c_1(c_1+1)\gamma}{2} &\leq& n \\
\ell &\leq& \frac{c_1(c_1+1)\gamma}{2}\\
\frac{n}{c_1}&\leq& \frac{c_1(c_1+1)c_2n}{2} \\
2n &\leq& (c_1^3 + c_1^2)c_2n \\
\frac{2}{c_1^3+c_1^2} &\leq& c_2 \ .
\end{eqnarray*}
Thus, all set families of size $c_1+1$ or greater must have $c_2 \geq \frac{2}{c_1^3+c_1^2}$, and  $c_2 < \frac{2}{c_1^3+c_1^2}$ guarantees the set family will have a size of at most $c_1$.

Since $c_1\ell = n$, it is possible to make a family of size $c_1$ for any value of $c_2$ by simply choosing sets that share no elements. Therefore, the size of the largest possible set family for any $n, \ell, \gamma$ meeting the specified conditions is $c_1$ if  $c_2 < \frac{2}{c_1^3+c_1^2}$.

If $c_2 \geq \frac{2}{c_1^3+c_1^2}$, there will always exist a set family of size $\geq c_1 +1$. To create such a family, choose $c_1 + 1$ sets such that each of them shares $\gamma$ elements with each of the others. This will be possible as long as: 
\begin{eqnarray*} 
c_1\gamma&\leq& \ell \\ 
nc_1c_2&\leq& \frac{n}{c_1} \\
c_1^2c_2 &\leq& 1 \\
\frac{2c_1^2}{c_1^3+c_1^2} &\leq& 1 \ .
\end{eqnarray*}
Since this final inequality is true for all possible values of $c_1$, it will such a set family can always be created, and its size will be, as shown earlier, $n$ when $c_2 = \frac{2}{c_1^3+c_1^2}$. Since increasing $c_2$ will not eliminate any possible set families, no $n, \ell, \gamma$ satisfying the conditions with $c_2 \geq \frac{2}{c_1^3+c_1^2}$ will have a maximum family size $<c_1 + 1$. Therefore, the size of the largest possible set family for a valid $n, \ell, \gamma$ will be $c_1$ iff $c_2 < \frac{2}{c_1^3+c_1^2}$. \end{proofof}

\cut{
\begin{theorem} \label{thm:UpperBound} \cite{NaturallyRehearsingPasswords}
Suppose that $\mathcal{S} = \left\{S_1,...,S_m\right\}$ is a $\left(n,\ell,\gamma\right)$-sharing set family of size $m$ then $m \leq {{n \choose {\gamma+1} }} \Big/ {{\ell \choose {\gamma+1} } }$.
\end{theorem}
}

The proof of theorem \ref{thm:chernoffapplication} is based on standard concentration bounds due to Chernoff. We use the specific form from Theorem \ref{thm:chernoffbound}. We demonstrate the existence of an $(n, \ell, \gamma)$--sharing set family of exponential size by showing that the probability of obtaining such a set family through random selection is non-zero.

\begin{theorem} \label{thm:chernoffbound} \cite{chernoff1952measure} Let $X_1, \ldots , X_n \in [0,1]$ be a sequence of independent random variables. Let $S = \sum_{i=1}^n x_i$, and let $\mu = \mathbf{E}[S]$. Then for all $\delta \geq 0$  \[ \Pr [S \geq  \mu + \delta n] \leq e^{-2n\delta^2 }\ .\]
\end{theorem}

\begin{remindertheorem}{\ref{thm:chernoffapplication}}
\thmChernoffApplication
\end{remindertheorem}

\begin{proofof}{Theorem \ref{thm:chernoffapplication}} We create an $(n, \ell, \gamma)$--sharing set family by creating sets in the following manner: Independently choose random integers $r_i^j$ each in the range $0\leq r_i < c_1$ for $j \in [m]$ and $i \in \{0,\ldots,\ell-1\}$. Let $S_j = \bigcup\limits_{i=0}^{\ell-1} \left\{ ic_1 + r_i^j\right\}$. Given two such sets, $S_j, S_k$ let 
\[x_i = \left\{
     \begin{array}{lr}
       1 & : r_i^j = r_i^k\\
       0 & :  r_i^j \neq r_i^k
     \end{array}
   \right.\] Then the number of elements shared by $S_j$ and $S_k$ is \[S_j \cap S_k = \sum\limits_{i=0}^{\ell-1} x_i \ .\] Let $\mu = \mathbf{E}\left[S_j\cap S_k\right]=\frac{n}{c_1^2}$ denote the expected number of shared elements. The probability that two such sets share more than $\gamma$ elements, given $c_2 = \frac{1}{c_1^2} + \epsilon$ is 
\begin{eqnarray*}
Pr[\left|S_j \cap S_k\right| > \gamma] &=& Pr[\sum\limits_{i=0}^{\ell-1} x_i > c_2n] \\
&=& Pr[\sum x_i > \frac{n}{c_1^2} + n\epsilon] \\
&\leq& Pr[\sum x_i \geq \mu + \epsilon n] \\
 &\leq& e^{-2n\epsilon^2}
\end{eqnarray*}
with the last step by Theorem \ref{thm:chernoffbound}. Thus the probability that two randomly selected sets share more than $\gamma$ elements is at most $e^{-2n\epsilon^2}$. 

An  $(n, \ell, \gamma)$--sharing set family of size $m$ will contain $m \choose 2$ pairs of sets. The probability that the family is valid, with none of the sets sharing more than $\gamma$ elements is 
\begin{eqnarray*}Pr[\exists j \neq k : \left|S_j \cap S_k\right| > \gamma] &\leq& {m \choose 2} Pr[\left|S_j \cap S_k\right| > \gamma] \\
&\leq& {m \choose 2} e^{-2n\epsilon^2} \\
&\leq& m^2 e^{-2n\epsilon^2}
\end{eqnarray*} 
by the union bound. For $m<e^{n\epsilon^2}$, this probability will be less than 1, meaning there is a non-zero chance of forming a valid set family of size $m$ by random selection and therefore such a family must exist.
\end{proofof}

\end{document}